\documentclass[12pt]{article}
\usepackage{epsfig,psfrag,verbatim,amssymb,amsfonts,amsthm}
\usepackage{graphicx}
\usepackage{subfigure,url,amsmath}
\usepackage{subfigure,amsmath,amssymb}
\usepackage{multirow}
\usepackage{color}
\usepackage{natbib}
\usepackage{vmargin} 
\usepackage{url}
\usepackage{threeparttable,tabularx,rotating,pdflscape}

\usepackage{arydshln}
\setpapersize{USletter}
\setmarginsrb{1truein}{1.1truein}{1truein}{1.3truein}{0pt}{0pt}{0pt}{0pt}
\newtheorem{lem}{Lemma}[section]

\def\tr{{\textrm tr}}

\def\L{\bar L}
\def\PL{P\bar L}
\newif\ifblind
\newif\ifunblind
\unblindtrue

\begin{document}
\begin{center}
{\large \bf Movement Prediction Using Accelerometers in a Human Population}\\[2ex]
{\small
\ifblind Anonymous Authors \else {\small {\sc Luo Xiao$\null^{1}$, Bing He$\null^{1}$, Annemarie Koster$^{2}$, Paolo Caserotti$^{3}$,
 Brittney Lange-Maia$\null^{4}$, Nancy W. Glynn$\null^{4}$, Tamara Harris$\null^{5}$  and Ciprian M. Crainiceanu$\null^{1}$}}\\[1ex]
{\it The Johns Hopkins University, Baltimore, MD, U.S.A.}$\null^{1}$\\
{\it University of Maastricht, Maastricht, The Netherlands}$\null^{2}$\\
{\it University of Southern Denmark, Odense, Denmark}$\null^{3}$\\
{\it University of Pittsburgh, Pittsburgh, PA, U.S.A.}$\null^{4}$\\
{\it National Institute on Aging, Bethesda, MD, U.S.A.}$\null^{5}$\\
}
\fi
\end{center}
\centerline{\today}
\begin{abstract}
We introduce  statistical methods for predicting  the types of human activity at  sub-second resolution using  triaxial accelerometry data. The major innovation is that we use labeled activity data from some  subjects to predict the activity labels of other subjects. To achieve this, we normalize the data across subjects by matching the standing up and lying down portions  of triaxial accelerometry data. This is necessary to account for differences between the variability in the position of the device relative to gravity, which are induced by body shape and size as well as by the ambiguous definition of device placement. We also normalize the data at the device level to ensure that the magnitude of the signal at rest is similar across devices.  After normalization  we use overlapping movelets (segments of triaxial accelerometry time series)  extracted from some of the subjects to predict the movement type of the other subjects. The problem was motivated by and is applied to a laboratory study of 20 older participants who performed different activities while wearing accelerometers at the hip.
Prediction results based on other people's labeled dictionaries of activity performed almost as well as those obtained using their own labeled dictionaries. These findings indicate that prediction of activity types for data collected during natural activities of daily living may actually be possible.
\end{abstract}

\section{Introduction}\label{sec::intro}
Body-worn accelerometers provide objective and detailed measurements of physical activity and  have been widely used in observational studies and clinical trials (\citealt{Atienza:05, Boyle:06, Bussmann:01,  Choi:11, Grant:08, Kozey:11, Schrack:13, Sirard:05, Troiano:08}). However, it is challenging to transform the accelerometry data into quantifiable and interpretable  information such as activity intensity or energy expenditure (\citealt{Bai:14, Schrack:14, Staudenmayer:12, Troiano:08,Trost:05, Welk:00}). 
An important goal of these studies is to transform an observed accelerometry dataset into a series of activity types that is time-stamped.  In this paper we are concerned  with predicting activity types at  sub-second resolution using detailed triaxial accelerometry information.
Sub-second labels seems to be the highest resolution that matters in terms of human activity recognition. Indeed, most human movements occur between 0.3 and 3.4 Hz (\citealt{Sun:93}).
Moreover, the resolution is necessary as we are interested in capturing short movements such as walking 2 or 3 steps, which is a highly prevalent type of activity in real life and likely to become a bigger component of activity as people age.
Such labeled time series data could then be used for health association studies, where decreases in activity diversity or changes in the circadian rhythm of activities may represent early strong indicators of biological processes or diseases. These expectations have strong face validity, as, for example, 1) an early indication of health recovery after surgery is the will and ability of a patient to use the bathroom; 2) disease may be associated with early reduction or abandonment of non-essential activities; and 3) death is associated with exactly zero activities.


\subsection{Accelerometry Data}\label{sec::intro::raw}

An accelerometer is a device that measures acceleration. When attached to the body of a human subject, if the subject is at rest, the accelerometer measures the subject's orientation relative to the gravitational vector; if the subject is moving, the accelerometer measures  a combination of the subject's orientation and acceleration.  Recent technology advances have produced  small and light accelerometers  that could collect data at high sample rates. For example, the Actigraph GT3+ device is of size $4.6\times 3.3\times 1.5 cm$,  weighs only 19 grams, and could sample data at 100 Hz; see Figure~\ref{orientation} for pictures of this device. Thus,  accelerometers can be easily attached to the human body and used for objectively recording detailed accelerations due to human physical activities.

\begin{figure}[htp]
\centering
\includegraphics[width = 4in,  angle=0]{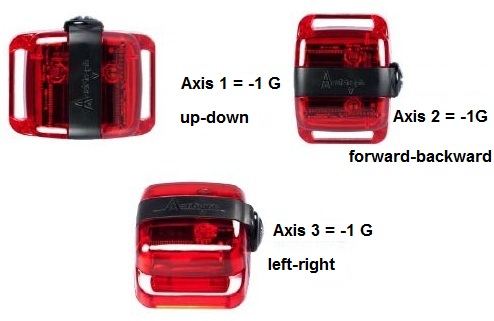}
\caption{\label{orientation} An Actigraph GT3+ accelerometer and its standard orientation. The top left, right and bottom graphs show the acceleration due to Earth's gravity when the corresponding axis aligns up with the opposite direction of  Earth's gravity. When the up-down axis is  the $x$-axis, the coordinate system is right-hand oriented. }
\end{figure}

The output from triaxial accelerometers, such as the Actigraph GT3+ device, is a triaxial time series of accelerations along three mutually orthogonal axes and expressed in units of Earth gravity, i.e., $g = 9.81m/s^2$. The three axes are labeled as ``up-down", ``forward-backward" and ``left-right" according to a device-specific reference system; see Section~\ref{sec::sub::raw} for more details. While the axes have these labels they have meaning only in the reference system of the device, as the device will move with the part of the body it is attached to. This means that an axis that is up-down relative to the device can easily be forward-backward, left-right, down-up or anything in between in the body-reference system. Figure~\ref{data} shows five segments of data for two subjects wearing the  devices at hip. From top to bottom, subjects perform 5 replicates of standing up from a chair and sitting back (chairStands), walk 20 meters at normal speed (normalWalk), mimic vacuuming,  stand still, and lie down on the back. From the data we observe that: 1) variability for active periods (normalWalk, chairStand and vacuum) is higher than inactive periods (standing and lying); 2) within each subject, the ordering and relative position of the three axes are different for standing and lying as  the orientation of the accelerometer with respect to the gravity differs for the two postures. Observations indicate that accelerometers are capable of detecting and differentiating various human activities. Moreover, we see that for chairStand and normalWalk, data for both subjects exhibit rhythmic patterns, suggesting that within the subject movements for same activities appear similar in the accelerometry data.

\subsection{Motivating Data}
The motivating data were collected from 20 older adults who were originally enrolled in the Study of Energy and Aging (SEA) pilot study. These participants were invited for an ancillary study for validating hip and wrist accelerometry and were instructed  in a research clinic to perform 15 different types of activities  according to a protocol. 
 Table~\ref{Table_types} provides the labels, detailed description and durations for the 15 activities.  The selection and design of these activities are intended to simulate a free-living context. The activity types are referred to by their labels in the paper. Throughout the study, each participant wore three Actigraph GT3X+ devices simultaneously, which were worn at the right hip, right wrist and left wrist, respectively.  The data were collected at a sampling frequency of 80Hz. Based on the protocol and the start/end times for each activity, a time series of labels of activity types is constructed to annotate the accelerometry data. In this paper we will focus on the data collected from accelerometers located at the hip and study how well a given program of activities can be distinguished by the accelerometry data at the population level.

\begin{table}
\small
    \centering
     \def\~{\hphantom{0}}
     \begin{minipage}{165mm}
      \caption{ \label{Table_types}15 activity types: labels, detailed description and durations}
    \begin{tabular}{p{2.5cm} |p{2 cm} p{10cm} p{2.5cm} }
 \hline
    Groups & Labels & Description & Duration \\
  \hline
\multirow{2}{1.5cm}{Resting}	&lying	&lay still face-up on a flat surface with arms at sides and legs extended	&10 mins\\
	&standing	&stand still with arms hanging at sides	&3 mins\\
    \hline 
\multirow{8}{2.5cm}{Upper body (while standing)}	& washDish	& fetch wet plates from a drying rack, dry them using a trying towel, and stack adjacent to the drying rack one-by-one	 &3 mins\\
	&knead	&knead a ball of playdough as if for cooking/baking	&3 mins\\
	&dressing	&unfold lab jacket, put jacket on (no buttoning), then remove, place the jacket on a hanger, and put the hanger on a nearby hook	 &3 mins\\
	&foldTowel	&fold towels and stack them nearby	&3 mins\\
	&vacuum	&vacuum a specified area of the carpet	&3 mins \\
	&shop	&walk along a long shelf, remove labeled items from the upper shelf about chest height, and place them on the lower shelf about waist height	&3 mins\\
\hline
\multirow{3}{2.5cm}{Upper body (while sitting)} 	&write &	write a specified sentence on one page of the notebook, then turn to the next page and repeat	&3 mins\\
	&dealCards	&hold a full deck, and deal cards one-by-one to six positions around a table	&3 mins\\
\hline
\multirow{7}{3.5cm}{Lower body} &	chairStand	&starting in a sitting position, rise to a normal standing position, then sit back down	&5 cycles\\
	&normalWalk&starting from standing still, walk 20 meters at a comfortable pace	 & 20 meters\\
	&normalWalk NoSwing	&starting from standing still, walk 20 meters at a comfortable pace with arms folded in front of chest	& 20 meters\\
	&fastWalk 	&starting from standing still, walk 20 meters at the fastest pace	& 20 meters\\
	&fastWalk NoSwing	&starting from standing still, walk 20 meters at the fastest pace with arms folded in front of chest	& 20 meters\\
    \hline
        \end{tabular}
             \end{minipage}
\vspace*{6pt}
\end{table}

We revisit Figure~\ref{data}, which displays the raw accelerometry data obtained from the hip accelerometers. We focus on the data for chairStand and normalWalk,  which exhibit rhythmic patterns.  An important observation is that these repetitive movements look very similar within the same person, though not across persons; this is a crucial observation as most prediction techniques are fundamentally based on similarities between signals. 
For example, for chairStand,  sudden large changes in acceleration magnitudes can be observed in the left-right axis for subject 13 but not for subject 4. Another example is that for normalWalk accelerations along the up-down and forward-back axes align up very well for subject 4 but are far apart for subject 13. These dissimilarities  seem to suggest that the accelerometry data are not comparable across subjects. Simply throwing prediction techniques at such a problem, irrespective to how sophisticated or cleverly designed they are, would achieve little in terms of understanding the data structure or solving the original problem. However,  we will show that a substantial amount of these observed dissimilarities is due to the orientation inconsistency of the reference systems across subjects and can be significantly reduced by using the same orientation, i.e., a common reference system. If a common reference system were used, then the three axes for standing and lying in Figure~\ref{data} would be very similar for the two subjects. This is clearly not the case for the raw accelerometry data as the left-right and the up-down axes for lying overlap for subject 13 but are quite different for subject 4.

\begin{figure}[htp]
\centering
\includegraphics[width = 6in, angle=0]{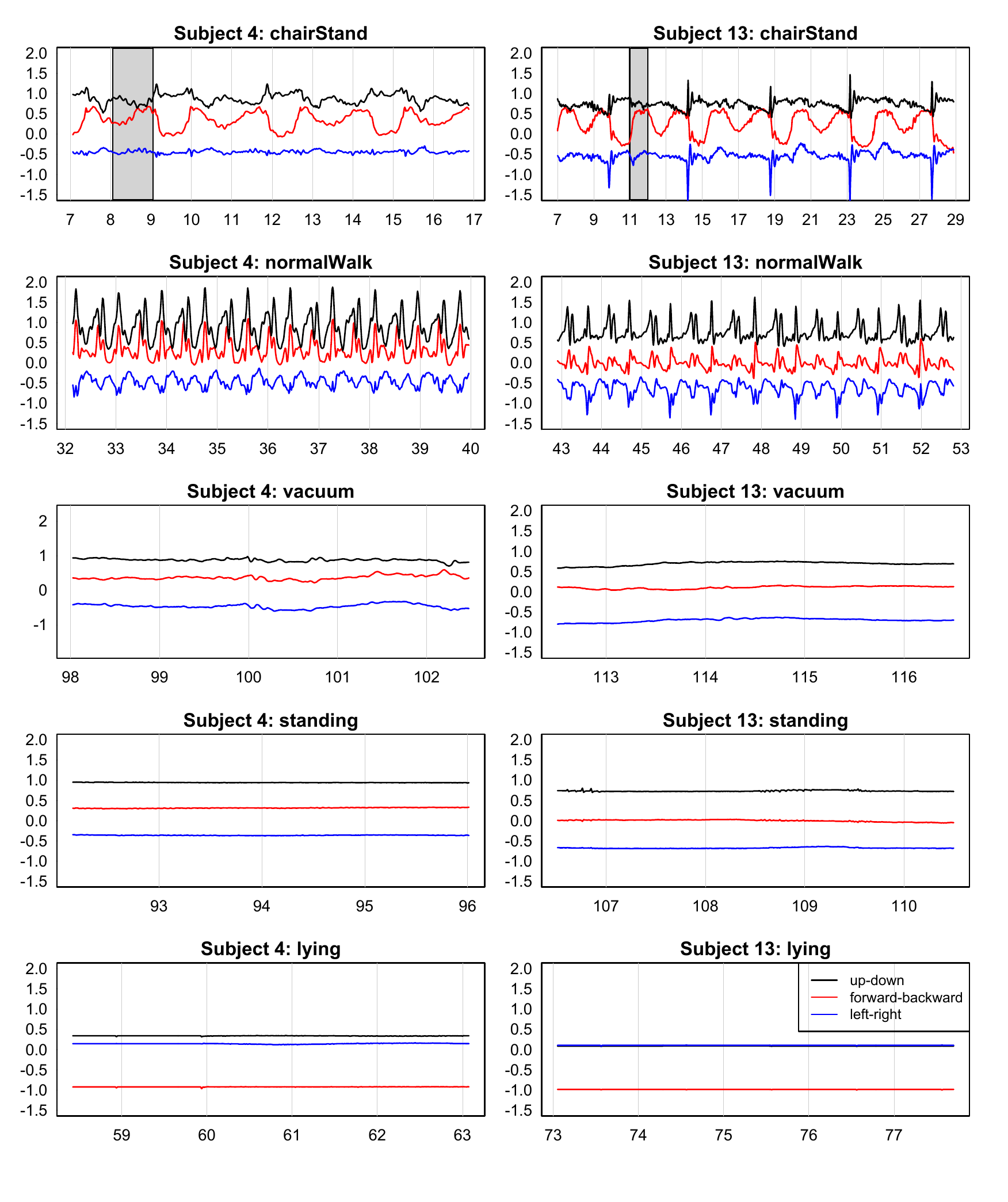}
\caption{\label{data} Raw data of chairStand, normalWalk, vacuum, standing and lying from hip-worn accelerometers for two subjects. The $x$-axis denotes recording time in seconds; the $y$-axis denotes the signal expressed in $g$ units.  The legend in the bottom right plot applies to all plots. The shaded areas contain two
$1s$ movelets. 
Note that orientation and placement of the device may change in the reference system designed around Earth's gravity.}
\end{figure}

\subsection{Proposed Methods}
In this paper, we  first address the problem that the raw accelerometry data collected from different subjects may not  be directly comparable. We show that the raw data are measured with respect to different reference systems and thus have different meanings across subjects. We will provide  a data transformation approach for normalizing the data, which is designed to mitigate these inherent problems in data collection.

Once data are normalized we proceed to predict activities using some of the subjects for training and the remaining subjects for testing the prediction algorithm. In particular, we will use movelets, a dictionary learning based approach that extends the methodology in \cite{Bai:12} designed for same-subject prediction. Here we describe briefly what  movelets are and provide the intuition behind the approach. A movelet is the collection of the three acceleration time series in a window of given size (e.g., $1s$). For example,  the time series in the two shaded areas in Figure~\ref{data} represent two $1$-$s$ movelets. The sets of overlapping movelets constructed from the accelerometry 
data with annotated labels are organized by activity types, which play the role of accelerometry ``dictionaries" for different activity types. As was illustrated
in Section~\ref{sec::intro::raw}, the intuition is that movements, and the associated movelets,  are similar for same activities and different for different activities. Based on this intuition, predictions of activity type based on accelerometry data without annotated labels can be obtained by identifying the annotated movelet that is most similar to the data; the similarity is quantified by the $L^2$ distance.  An important problem with the
movelets approach is choosing the window size for the movelets. We will introduce a criterion based on prediction performance, evaluate the criterion in our data, and provide specific recommendations for the optimal size of the movelets. This gives us a rigorous and  data-based approach that provides the necessary context for the currently used window size.  

\subsection{Existing Literature} 
A number of methods have been used for recognition of activity types, including linear/quadratic discriminant analysis (\citealt{Pober:06}), hidden Markov models (\citealt{Krause:03}), artificial neural networks (see, e.g., \citealt{Staudenmayer:09, Trost:12}), support vector machines (see, e.g., \citealt{Mannini:13}) and combined methods (see, e.g., \citealt{Zhang:12}). \cite{Bao:04} and \cite{Preece:09} reviewed and evaluated methods used in classifying normal movements. The major limitations of
these methods include: 1) they usually require at least a 1-minute window to conduct feature extraction; 2) they do not capture finer movements that last less than 1 minute, such as falling or standing up from a chair; and 3) the prediction process is usually hard to understand and interpret. In contrast, our proposed method is fully transparent, easy to understand,  requires minimal training data, and is designed both for periodic and non-periodic movements.

The rest of the paper is organized as follows. In Section \ref{sec::preprocess}
 we consider two main factors for making the raw accelerometry data not comparable
 and propose a  transformation method for normalizing the data.
 In Section \ref{sec::method}, we describe in detail the movelet-based prediction methods.
 In Section \ref{sec::result} we apply our prediction method to some real data. We
 conclude the paper with a discussion about the feasibility of movelet-based movement prediction
 and its potential relevance to public health applications.

\section{Triaxial Accelerometry Data Normalization}\label{sec::preprocess}
\subsection{Raw Data Are Not Comparable Across Subjects}\label{sec::sub::raw}
A fundamental problem with the raw triaxial accelerometry data is that they may not be directly comparable across subjects; Figure \ref{data} provides the intuition about this problem and indicates that orientation inconsistency  is an important factor. Another factor is device-specific systematic bias, which we will explain in details later.

More specifically, triaxial accelerometers measure accelerations along the three axes in the reference system of the device. Indeed, Figure \ref{orientation} indicates exactly the device-specific reference system. In its standard orientation and in absence of movement, the up-down axis of the device is aligned with Earth's gravity and will register $-1$g (acceleration towards the center of the earth) and 0g along the other two axes (orthonormal to Earth's gravity).  If the device is rotated clockwise by 90 degrees towards the forward-backward (left-right) axis then in absence of movement the forward-backward  acceleration will  continuously record $-1$g, while the acceleration in the other two directions will be 0g. This shows that: 1)  orientation of the device will fundamentally alter the local mean of the acceleration; 2) the relative size of the acceleration along the axes is a proxy for orientation of the device relative to Earth's gravity; and 3) simple laws of physics may be applied to make the signals comparable across individuals by using a common reference system. Our paper is dedicated to showing that using a common reference system is a crucial step for across-individual prediction of movement type, and providing a simple way to predict across individuals using the dictionary of movement of others (movelets).  In some applications, accelerometers are mounted on a rigid body to maintain the standard orientation. However,  when worn by human subjects, the orientations of accelerometers  could depend on many factors. First, the device can be flipped,  turned or attached to a slightly different part of the human body. Another factor is the shape of the  human body that the device is attached to; for example, waist circumference or relative height of the hip with respect to the body can fundamentally alter what we generally describe as ``position".  To illustrate the problem, we have extracted two segments of data ($3 s$ for each segment) for standing and lying for each of the 20 subjects in our study and computed the means of three axes for each segment; see the left panel of Figure~\ref{centers_o}. If the orientation of the accelerometers respected our intuition of directions then the means should be $(-1,0,0)$ for standing and $(0,-1,0)$ for lying. Here we use a length-3 vector  to denote the acceleration in the order of  up-down, forward-backward, and left-right axes. The left panel of Figure~\ref{centers_o} indicates that the orientations of the accelerometers are rarely they are expected to be  and  vary considerably across subjects. It is interesting to see that almost all accelerometers were flipped in the up-down direction because for those accelerometers the means of data for standing are $(1,0,0)$; this probably indicates that even something as simple as up-down relative to the device can be easily misinterpreted.  
An example of device placement that is inconsistent with the protocol is  subject 10, for whom the mean of data for standing is $(0,0,1)$. The means for subjects 13 and 15 show strong tilt/rotation effects. These  simple exploratory tools indicate the need of rotating the data to ensure that they have about the same interpretation across subjects; we believe that the lack of appreciation of this fact has been the major obstacle to current attempts to predict movement types across individuals.

\begin{figure}[htp]
\begin{center}$
\begin{array}{cc}
\includegraphics[width = 3in, angle=0]{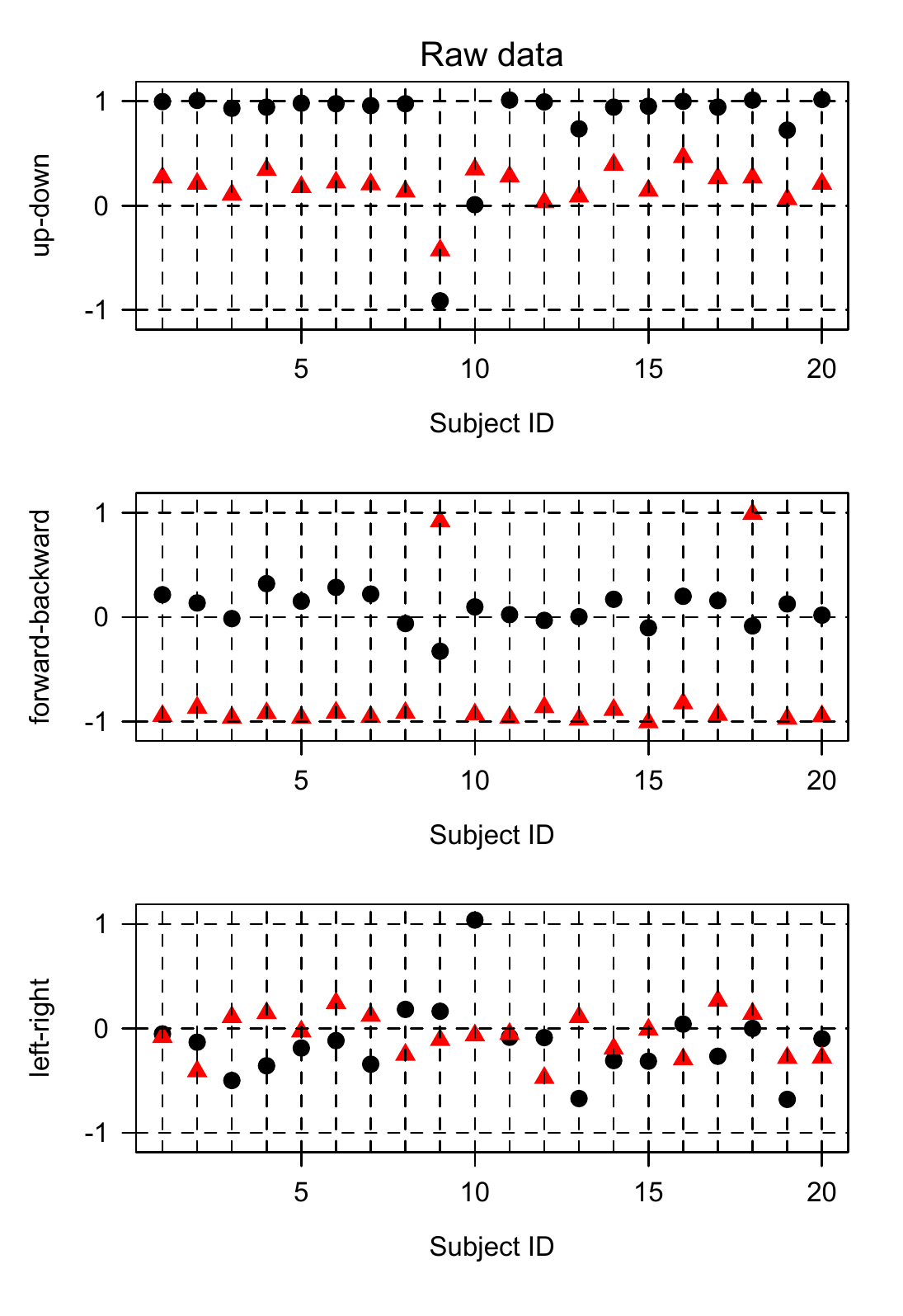}&\includegraphics[width = 3in, angle=0]{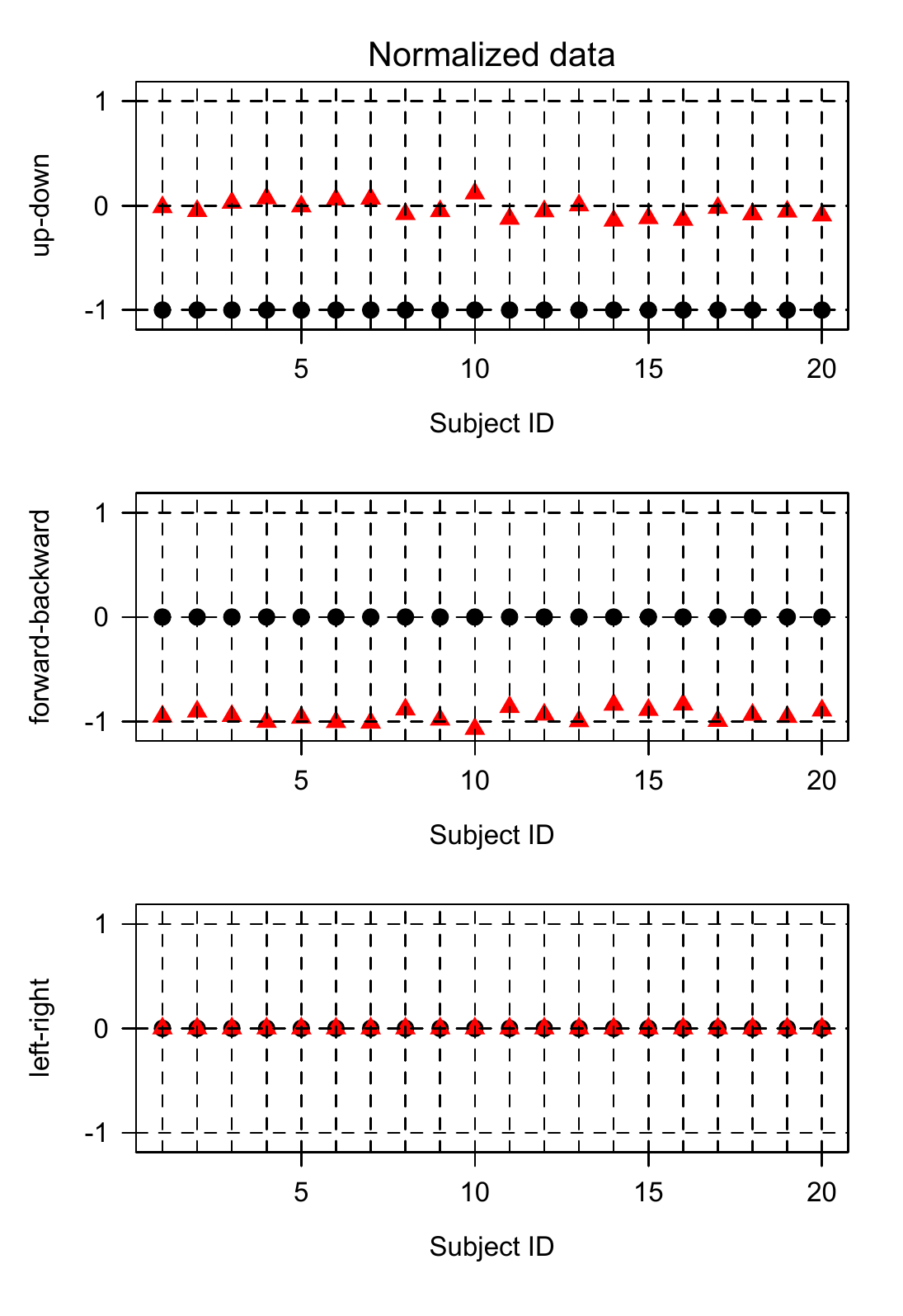}
\end{array}$
\end{center}
\caption{\label{centers_o} Means of standing and lying using the raw data (left panel) and the normalized data (right panel). Filled circles are for standing while filled triangles are for lying. The $x$-axis denotes subject ID; the $y$-axis is $g$ units. The top panel is for the up-down axis, the middle one is for the forward-backward axis, and the bottom one is for the left-right axis. This figure appears in color in the electronic version of this article. }
\end{figure}

Another problem is that  the accelerometry data  have small but non-ignorable systematic biases; more precisely, the magnitude of the acceleration vector differs when it should not, that is, when subjects stand still.  To separate the bias  from the orientation inconsistency, we have computed the magnitude of the mean acceleration vector, which is $\sqrt{x_1^2 + x_2^2 + x_3^2}$ if the mean acceleration vector is denoted by $(x_1, x_2, x_3)^T$. Note that the magnitude of the mean acceleration vector is invariant to rotation. Figure~\ref{magnitude} plots the magnitude of the mean acceleration vector for standing, which should be $1g$. An inspection of the graph indicates that the magnitude is rarely equal to 1 and  it differs from 1 by as much as 6\% for two subjects. We have found that these differences can have serious consequences for activity type prediction, because a change of this magnitude can fundamentally affect the geometry of the activity space. Using a simple model where we assume that subjects have random moves while  standing still (see Appendix A) we could show that the gravity-inflation (note that magnitudes are typically larger than 1) is most likely not due to random movements for most devices. Thus,  for the purpose of our research we will treat these differences as systematic biases that are associated with the devices. A more in-depth and principled analysis of this assumption is beyond the scope of the current paper.

\begin{figure}[h]
\centering
\includegraphics[width = 4in, angle=0]{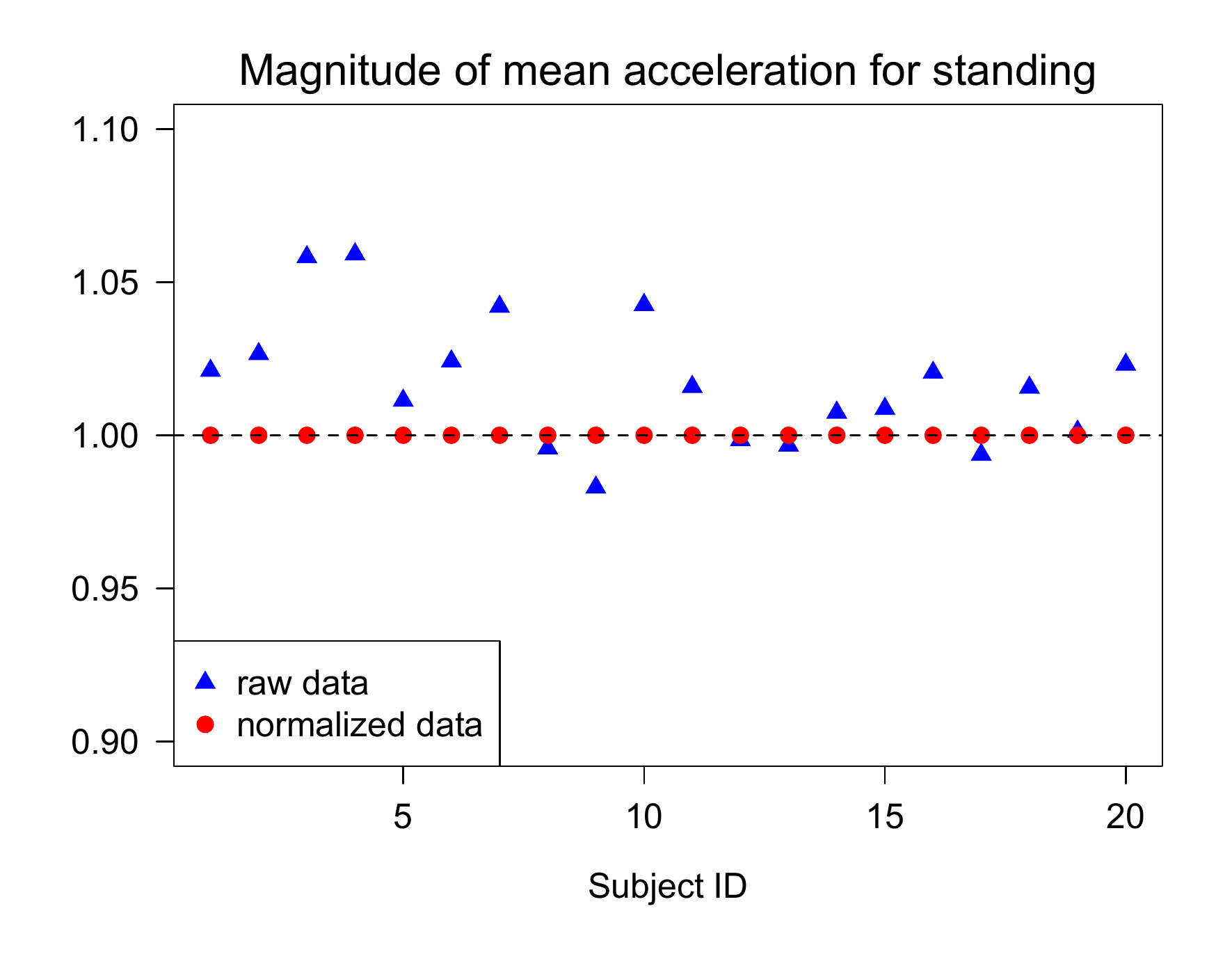}
\caption{\label{magnitude} Magnitude of the mean acceleration vector during standing (expressed in $g$ units) using the raw and normalized data. Blue triangles are for raw data while red circles are for normalized data.}
\end{figure}

\subsection{Normalization of Raw Triaxial Data}\label{sec::sub::normalization}
The purpose of normalization is to make  same-activity data  more comparable across subjects so that we can use the dictionary of movements of one or several subjects (movelets) to predict the activity of others. 
To achieve this, a desired normalization procedure should be able to correct the orientation inconsistency and  reduce the bias  in the raw data.

The proposed procedure is to apply a particular linear transformation to the raw data. This consists of two steps: rotation and translation. We first rotate the triaxial data so that all the data are in the standard orientation (reference system) and then translate the rotated data to reduce  systematic biases. To be precise, assume $\mathbf{u} = (u_1, u_2, u_3)^T$ is  a single data point in the raw data space, $\mathbf{R}$ is a rotation matrix, and $\mathbf{b}= (b_1,b_2,b_3)^T$ is the vector of systematic bias. Then we have
\begin{equation}
\label{transform}
\mathbf{x} = \mathbf{R}\mathbf{u} - \mathbf{b}.
\end{equation}
In practice, accelerometers might be moving  from time to time relative to the  human body, which could make their orientation  time-dependent. We assume that the accelerometers do not move with respect to the human body and apply the same transformation in~(\ref{transform}) to all  raw data within the same subject. This simple approach has worked well in practice.

We need to determine $\mathbf{R}$ and $\mathbf{b}$ based on the subject-specific raw data; both $\mathbf{R}$ and $\mathbf{b}$  depend on the subject, but the notation was dropped for presentation simplicity. 
We extracted two small segments from the raw accelerometry data, one segment for standing and  another for lying; the segments can be very short, say 2-3 seconds. 
The approach then proceeds by calculating the means of the three axe  for both activities, which results in two three dimensional vectors $\mathbf{a}_1$ for standing and $\mathbf{a}_2$ for lying. If the up-down axis for the orientation of the device aligns  well with the negative direction of  Earth's gravity and the data  have no systematic bias, then $\mathbf{a}_1$ should be close to $-\mathbf{e}_1 = (-1,0,0)^T$ and  $\mathbf{a}_2$ should be close to  $-\mathbf{e}_2 = (0,-1,0)^T$. Hence,  we select $\mathbf{R}$ from the class of rotation matrices that minimizes $\|\mathbf{R}\mathbf{a}_1 +\mathbf{e}_1\|_2^2 + \|\mathbf{R}\mathbf{a}_2 + \mathbf{e}_2\|_2^2$ and satisfies $\mathbf{e}_3^T\mathbf{R} (\mathbf{a}_1\times \mathbf{a}_2)>0$, where $\mathbf{e}_3 = (0,0,1)^T$ and $\mathbf{a}_1\times \mathbf{a}_2$ is the cross product of $\mathbf{a}_1$ and $\mathbf{a}_2$. The latter condition  ensures that we  have a right-hand coordinate system for the rotated data. It can be shown that $\mathbf{R}$ is uniquely determined and can be computed, as shown in  Web Appendix C. Then we let $\mathbf{b} = \mathbf{R}\mathbf{a}_1+ \mathbf{e}_1$,  which by~(\ref{transform}) implies that the mean of standing accelerometry is exactly $-\mathbf{e}_1$. 

We note that estimation of the rotation matrix $\mathbf{R}$ might be affected by the existence of  systematic bias. For our accelerometry data, the systematic bias is small (see Figure~\ref{magnitude}) and seems to have negligible effect on rotation. However, if one is concerned about this issue, then we suggest a visual examination of the raw and normalized data, which might reveal if  normalization provides reasonable results.

\subsection{Normalized Data}
We applied the normalization method to the accelerometry data and re-calculated the means of standing and lying for the 20 subjects. As expected,  the right panel of Figure~\ref{centers_o} shows that the means of standing (lying) are close to $-\mathbf{e}_1$ ($-\mathbf{e}_2$) for all subjects. Comparing the two panels in Figure~\ref{centers_o} we see strong indications  that  the normalized data are  more comparable across subjects. The fact that means are closer to what they are designed to be after normalization is satisfying, though not surprising. To show the dramatic effects of normalization we investigate the changes from the raw to normalized data in activities that were not used for normalization. The top panels of Figure~\ref{compare_raw_transform} displays the raw data for normal walking for 2 subjects. A close inspection of the two graphs indicate periodic movements, though the patterns and the size of the means of movement make it very hard for any method developed for one of the subjects to recognize the patterns of the other subject. After normalization the patterns are much more comparable (see bottom panels in Figure~\ref{compare_raw_transform}), which will allow powerful techniques, such as movelets, to be generalized across individuals. It is interesting that while the signal is visually similar, the amplitude of the up-down axis (the black solid line in the bottom panels in Figure~\ref{compare_raw_transform}) differs quite substantially between the two subjects. This is likely due to the stronger up-down acceleration of subject 4 compared to subject 13. 

\begin{figure}[htp]
\centering
\includegraphics[width = 6in, angle=0]{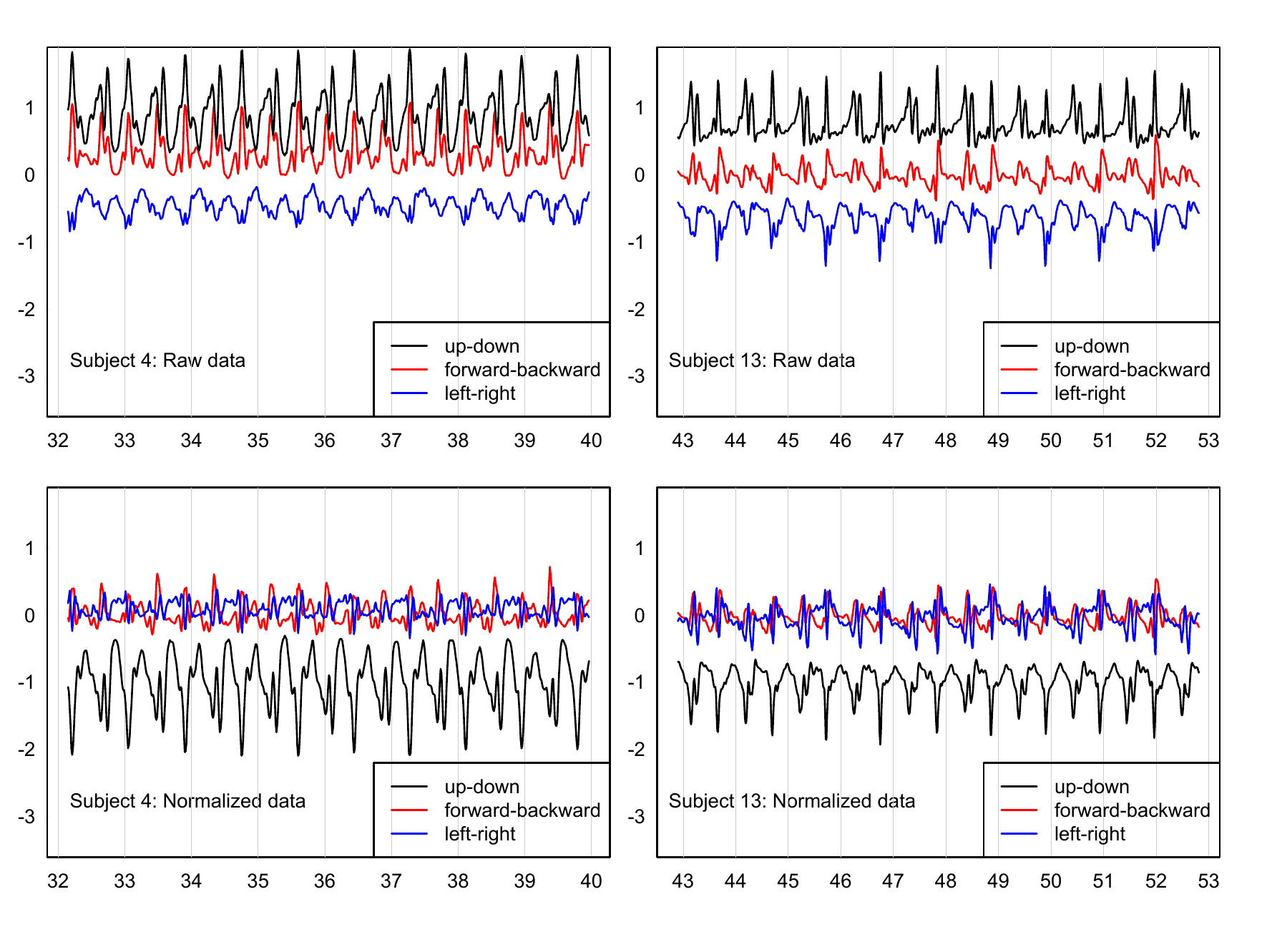}
\caption{\label{compare_raw_transform} Raw  and normalized data from two subjects for normalWalk. }
\end{figure}

\section{Movelets Prediction}\label{sec::method}
In this section, we first describe the subject-level movelets prediction method developed by \cite{Bai:12}  and introduce some notation. Then we propose a population-level movelets prediction method that could predict a subject's activity labels with no knowledge about how this subject's accelerometry data look  when doing various activities. Finally, we provide a simple automatic approach for tuning the window size/length of movelets, an important component in the movelets prediction procedure.

\subsection{Subject-level Movelets Prediction}
We denote the normalized triaxial accelerometry data by $\mathcal{X}_i(t) = \{X_{i1}(t), X_{i2}(t), X_{i3}(t)\}$ where $t\in \mathbb{T}_i$ and $i\in \mathbb{I}$. Here $\mathbb{T}_i$ is a time domain on the scale of seconds and $\mathbb{I}$ denotes the collection of subjects. Associated with the data is a label function $L_i(t)$ which takes values in $\mathcal{A}= \{Act_1,\dots, Act_A\}$, a collection of labels  each denoting a distinct activity. If  subject $i$ is standing at time $t$, $L_i(t)$ will be the label in $\mathcal{A}$ that represents standing. Let $\mathbb{U}_i$ and $\mathbb{W}_i$ be a disjoint union of $\mathbb{T}_i$. Then the subject-level prediction is to  predict the labels $\{L_i(t): t\in \mathbb{W}_i\}$ for the data $\{\mathcal{X}_i(t): t\in \mathbb{W}_i\}$, assuming $\{\mathcal{X}_i(t): t\in\mathbb{U}_i\}$ is the training data with $\{L_i(t): t\in \mathbb{U}_i\}$ being known. For the subject-level prediction, data and label information across subjects are not used.

The idea of movelets prediction (\citealt{Bai:12, He:14}) is to first decompose the accelerometry data into a collection of overlapping movelets. A movelet  of length $h$ at time $t$ is defined as $M_i(t,h) = \{\mathcal{X}_i(s): s\in [t, t+h)\}$ and it captures the acceleration patterns in the time interval $[t, t+h)$. For simplicity we drop the parameter $h$ from $M_i(t,h)$ hereafter. The accelerometry data can be decomposed into a continuous sequence of overlapping movelets with $\mathcal{X}_i(t)$ being contained in all the movelets starting at $s \in (t-h, t]$.  The approach of overlapping movelets, as a type of  sliding window technique, is advantageous  to other windowing techniques such as event-defined windows or activity-defined windows (\citealt{Preece:09}). The latter windowing techniques require either locating specific events or determining the times at which the activity changes; but the use of overlapping movelets requires no such ``preprocessing" and is better for real-life applications. Another advantage of overlapping movelets is that they are computationally easy and simple to construct, which makes them well suited for analyzing ultra-dense acceleration data.

Using movelets, the training data can be represented as $\{M_i(t): t\in \bar{\mathbb{U}}_i\}$ and similarly for the unlabeled data, $\{M_i(t): t\in \bar{\mathbb{W}}_i\}$. Here $\bar{\mathbb{U}}_i = \{t\in \mathbb{U}_i: [t, t+h)\subset \mathbb{U}_i\}$ and $\bar{\mathbb{W}}_i = \{t\in \mathbb{W}_i: [t, t+h)\subset \mathbb{W}_i\}$. Note that in the new data representation with movelets,   a data point will be lost if it is not contained in any  time interval of length at least $h$ in $\mathbb{U}_i$ or $\mathbb{W}_i$. To avoid ambiguity for prediction, the movelets in the training data that do not belong exclusively to a single type of activity are  deleted. For the movelets in the training data, we  define the label function $\L_i\{M_i(t)\} := L_i(t)$.

Given an unlabeled movelet $M_i(t), t \in \bar{\mathbb{W}}_i$, the method is to find the closest match in the training data, i.e., to search for  $t^{\ast} = t^{\ast}(t) \in \bar{\mathbb{U}}_i$  such that
\begin{equation}
\label{objective}
M_i(t^{\ast}):= \text{arg} \min_{s\in \bar{\mathbb{U}}_i} f\{M_i(t), M_i(s)\},
\end{equation}
where $f(\cdot, \cdot)$ is some function that measures  distance between  movelets; this type of 1-nearest neighbor works exceptionally well in conjunction with the sliding window movelets. The idea is that movelets with similar pattern or shape should belong to the same activity. A simple  $f$ that will be adopted in the paper is 
\begin{equation}
\label{f}
f\{M_i(t), M_i(s)\}: = h^{-1}\int_0^h\left \|\mathcal{X}_i(t+h^{\prime}) - \mathcal{X}_i(s+h^{\prime})\right\|^2 dh^{\prime}.
\end{equation}
 The above $L^2$ distance measures explicitly  the amplitude difference between two movelets; in addition, because movelets are always of the same length, the distance between movelets from activities of different frequencies can also be large. 


Since $M_i(t^{\ast})$ is the best match for $M_i(t)$, the label of the movelet at time $t$  can be predicted by $\PL_i (M_i(t)): = \L_i(M_i(t^{\ast})) = L_i(t^{\ast})$. We may use the predicted label for $M_i(t)$  as the prediction for $\mathcal{X}_i(t)$. However, this may not always be an accurate prediction because we use data as far as $h$ seconds away in making the prediction with the underlying assumption that the activity type remains the same in the time interval $[t, t+h)$. So if the movelet actually contains some transitional period between different activity types, this prediction can be wrong. Using the facts  that $\mathcal{X}_i(t)$ contributes to the prediction of all the movelets that  start at $s\in (t-h, t]$ and that human physical movements are in general continuous, the predicted label can be  the most frequently predicted label for all the movelets that contain $\mathcal{X}_i(t)$. Formally we let
\begin{equation*}
\label{prediction}
\begin{split}
PL_i(t):&=\text{arg} \max_{Act_j\in \mathcal{A}} h^{-1}\int^t_{t-h} 1_{\{ \PL_i(M_i(s)) = Act_j\}}ds\\
& = \text{arg} \max_{Act_j\in\mathcal{A}} h^{-1}\int^t_{t-h} 1_{\{ L_i(t^{\ast}(s))=Act_j\}} \mathrm{d} s,
\end{split}
\end{equation*}
where $1_{\{\cdot \}}$ is 1 if the statement inside the bracket is true and 0 otherwise.

\subsection{Population-level Movelets Prediction}
The movelets method described in the previous section  considers only prediction at the subject level as the training data and the  data to be predicted are from the same subject. We now extend the movelets prediction method to the population level. We assume $\mathbb{I}$, the collection of subjects, are divided into two disjoint sub-collections, $\mathbb{I}_0$ and $\mathbb{I}_1$. For subjects in $\mathbb{I}_0$ the activity labels are known while for these in $\mathbb{I}_1$ the labels are unknown. The problem is to predict the activity labels for subjects in $\mathbb{I}_1$ using the data from $\mathbb{I}_0$. 

Given an unlabeled movelet $M_i(t)$, $i\in \mathbb{I}_1$, the proposed method is to  search for $\{i^{\ast}, t^{\ast}(t)\}$ such that
\begin{equation}
\label{objective_pop}
M_{i^{\ast}}(t^{\ast}) := \text{arg}\min_{\left\{M_j(s):\  j\in \mathbb{I}_0, \ s\in \bar{\mathbb{T}}_j\right\}} f\{M_i(t), M_j(s)\},
\end{equation}
where $f$ is defined in~(\ref{f}) and $\bar{\mathbb{T}}_j = \{t\in \mathbb{T}_j: [t, t+h)\subset \mathbb{T}_j\}$. With small modifications, equation~(\ref{objective_pop}) reduces to~(\ref{objective}) when we consider  a single subject. The idea of the proposed method is that  we will be  able to label the  movelet accurately as long as it could match the movelet to a same-activity movelet of at least one subject in $\mathbb{I}_0$. This is the key for a successful  population-level prediction,  as movements from different subjects usually exhibit  different patterns and hence  movelets from different subjects have large within-activity variation. However, having a sizable group of subjects in the training set will cover multiple patterns in the normalized data, which will lead to improved prediction. In some sense, ours is the implementation of the intuition that ``many people move differently, but some people move like you". For example, for fast walking with arm swinging in the accelerometry data, the time for completing two steps ranges from about $0.375s$ to $0.625s$.  Hence, rather than requiring a subject's movement to be similar to the collection of movements of the same activity from all subjects in $\mathbb{I}_0$, our method only requires this subject's movement to be similar to at least one subject's movement of the same type. As long as there are people in the training dataset whose times for completing two steps are similar to the subject  we try to predict, prediction should work reasonably well.

Finally, the predicted label for $M_i(t)$ is  $\PL\{M_i(t)\}: = \L_{i^{\ast}}\{M_{i^{\ast}}(t^{\ast})\} = L_{i^{\ast}}(t^{\ast})$ and we let
\[
PL(i,t) := \text{arg} \max_{Act_j\in \mathcal{A}} h^{-1}\int^t_{t-h} 1_{\{ \PL(M_i(s))  = Act_j\}}\mathrm{d} s.
\]
\subsubsection{Evaluation of Prediction Results}\label{sec::evaluation}
We evaluate the performance of the proposed prediction method  by defining the following two quantities: 
 for activity type $Act_j\in\mathcal{A}$ and subject $i\in \mathbb{I}_1$, the subject-specific true prediction rate is
defined as
\begin{equation}
\label{apr}
r_{ij} :=\frac{\sum_{t\in\bar{\mathbb{T}}_i, L_i(t) = Act_j} 1_{\left\{PL(i,t) = Act_j\right\}}}{\sum_{t\in\bar{\mathbb{T}}_i, L_i(t) = Act_j}1},
\end{equation}
and the corresponding false prediction rate is 
$$
w_{ij} :=\frac{\sum_{t\in\bar{\mathbb{T}}_i,PL(i,t) = Act_j}1_{\left\{L_i(t)\neq Act_j\right\}}}{\sum_{t\in\bar{\mathbb{T}}_i,PL(i,t) = Act_j}1}.
$$
Note that $r_{ij}$ is the proportion of subject $i$'s data points in activity type $Act_j$ that are correctly identified as belong to $Act_j$, while $w_{ij}$ is the proportion of subject $i$'s data points that are identified as in activity type $Act_j$ but do not belong to $Act_j$.
A successful prediction method should have high $r_{ij}$ and low $w_{ij}$.
\subsubsection{Selection of Movelets Length}\label{h}
An important problem in  movelets prediction is the selection of $h$, the length of movelets. \cite{Bai:12} and \cite{He:14}  noted that $h$ should be selected such  that  movelets  contain sufficient information to identify a movement while avoiding redundant information; based on the guideline both papers used $h=1s$.  This choice is based  on the reasonable observation that human movement happens on the $0.5-2$Hz scale, though no objective criterion has been explored so far. Here we propose a simple automatic approach for selecting $h$ using only the training data. Our approach is based on leave-one-subject-out cross validation applied to the training data.

Consider subject $i\in \mathbb{I}_0$. For  activity type  $Act_j\in\mathcal{A}$ of subject $i$, we calculate its true prediction rate  defined in~(\ref{apr}) with subjects in $\mathbb{I}_0-\{i\}$ as training data; denote the accurate prediction rate by $r_{ij}^*$. Then the average prediction accuracy over all subjects and all activity types is 
\begin{equation*}
\bar r^* =  (I_0 A)^{-1}\sum_{i\in \mathbb{I}_0, Act_j \in \mathcal{A}}\mathit{r}_{ij}^*.
\end{equation*}
 As $\bar r^*$ depends on $h$,  we  propose to select 
\begin{equation}
\label{h}
h^*:= \arg \max_h \bar r^*.
\end{equation}

\section{Results of Movelet Prediction}\label{sec::result}

We now apply the movelets prediction method to the 20 participants of the accelerometry validation study. As described in the Introduction, data from hip, left  and right wrist-worn accelerometers were collected. The participants were instructed to perform 15 activities with some resting breaks (3 minutes each break) between  activities. Activities were chosen specifically because they were representative of activities that older adults may commonly do in their daily lives; please see Table \ref{Table_types} for specific details. The resting breaks were removed from the data and the transitional periods between consecutive activities were also removed. We focus on applying our method to data from hip-worn
accelerometers.

The raw data were first normalized using the  transformation proposed in Section~\ref{sec::preprocess}. The 20 participants  were then split into two groups: 10 for training  and 10  for testing. Because the accelerometry data for each subject contain dense triaxial time series, it will become computationally challenging if all of the training data are used for prediction.  For activities with explicit starting and ending, such as chairStand, two consecutive replicates from each subject were used as training data. For other activities, a segment of 5 seconds for each subject were used as training data. The shortness of the training periods is a hallmark of the movelets prediction approach. Indeed, movelets only need a quick look at some quality data to recognize the data in a complex system. 
\subsection{Activity Groups}
According to Table \ref{Table_types}, there are 8 types of upper body activities in the data. For prediction, we will group these activities into a new subgroup, ``upper body activities". There are two reasons for doing this. First, these  activities involve mainly movements from the upper body such as arms and hands and  are not well detected by accelerometers worn on the hip (\citealt{He:14}). Second, it was shown in \cite{He:14} that  upper body activities require a series of distinct sub-movements that can be similar across activities and difficult to differentiate even within the same subject. Differentiating these upper body activities across subjects  becomes even more challenging because of the increased heterogeneity of activities across subjects.

For the lower body activities, there are four types of walking activities: normal walking without arm swinging, normal walking with arm swinging, fast walking without arm swinging, and fast walking with arm swinging.  Arm swinging belongs to upper body activities and  is not  well detected by accelerometers worn on the hip (\citealt{He:14}).  Although distinguishing normal and fast walking can be done relatively easily when training data are available for the same subject (\citealt{He:14}),  it will be difficult  to do so across subjects. The reason is that there is a lot of heterogeneity in the subjects' walking speed and  one subject's fast walking speed may well be close to another subject's normal walking speed. Indeed, for the 10 subjects used as training data, the time for  two normal-walking steps ranges from $0.75 s$  to $1.25 s$ while it is between $0.375 s$ and $0.625 s$ for two fast-walking steps.  Thus,  a  new subgroup  of activities, ``walking", is created to  include all four walking types.

We now have 5 new activity types: ``standing", ``lying", ``chairStand", ``walking", and ``upper body activities". We will use the new activity labels for evaluating our population-level movelets prediction method.
\subsection{Selection of Movelets Length}\label{sec::sub::length}
We selected the movelet length $h=0.75 s$ by the criteria in~(\ref{h}); see Figure~\ref{apc_h} for cross-validated  mean  prediction accuracy for different choices of $h$. The high  cross-validated mean prediction accuracy (near $90\%$) indicates  that the population-level movelet method is capable of 
identifying activity types across subjects.
\subsection{Prediction Results}
We used the true prediction rate and false prediction rate described in Section~\ref{sec::evaluation} to evaluate the results. To illustrate the importance of  the preprocessing step, we also applied  the proposed method to the raw data and to the rotation-only data. As a comparison, the results by the within-subject movelets method applied to the 10 testing subjects are also shown. Note that for the within-subject movelets method, $2$-$3$ $s$ of labeled accelerometry data for each activity are used for predicting the activity type of the rest of the accelerometry data.
The top panel of Figure~\ref{fig_apr} shows box plots of  the true prediction rates for the 5 activity types and the bottom panel  shows box plots of  the false prediction rates. The two figures indicate that  the population-level  method for the normalized data performs similarly with the within-subject method, i.e., having high mean true prediction rates ($>90\%$) and low mean false prediction rates ($<20\%$). For the population-level method, there is larger variability in both true and false prediction rates across subjects, reflecting the increased uncertainty due to heterogeneity of same activities across subjects.  Moreover, the results demonstrate that a naive application of the population-level movelets method  to the raw data could lead to low true prediction rates and high false prediction rates. In particular, standing can not be identified across subjects. The role of translating the rotated data is also non-trivial (compare the results for  ``rotated" and ``normalized"). Indeed, the prediction performance is improved for chairStand and  for distinguishing other activities from chairStand.

To illustrate that the four types of walking are essentially indistinguishable across subjects for accelerometers worn on the hips, we plotted the prediction results for one subject's data for walking; see Figure~\ref{plots_walk}. The results show that this subject's fast walking with arm swinging 
is identified as fast walking without arm swinging and his or her fast walking without arm swinging is identified as normal walking without arm swinging.

\begin{figure}[htp]
\begin{center}$
\begin{array}{c}
\includegraphics[width=6in,  angle=0]{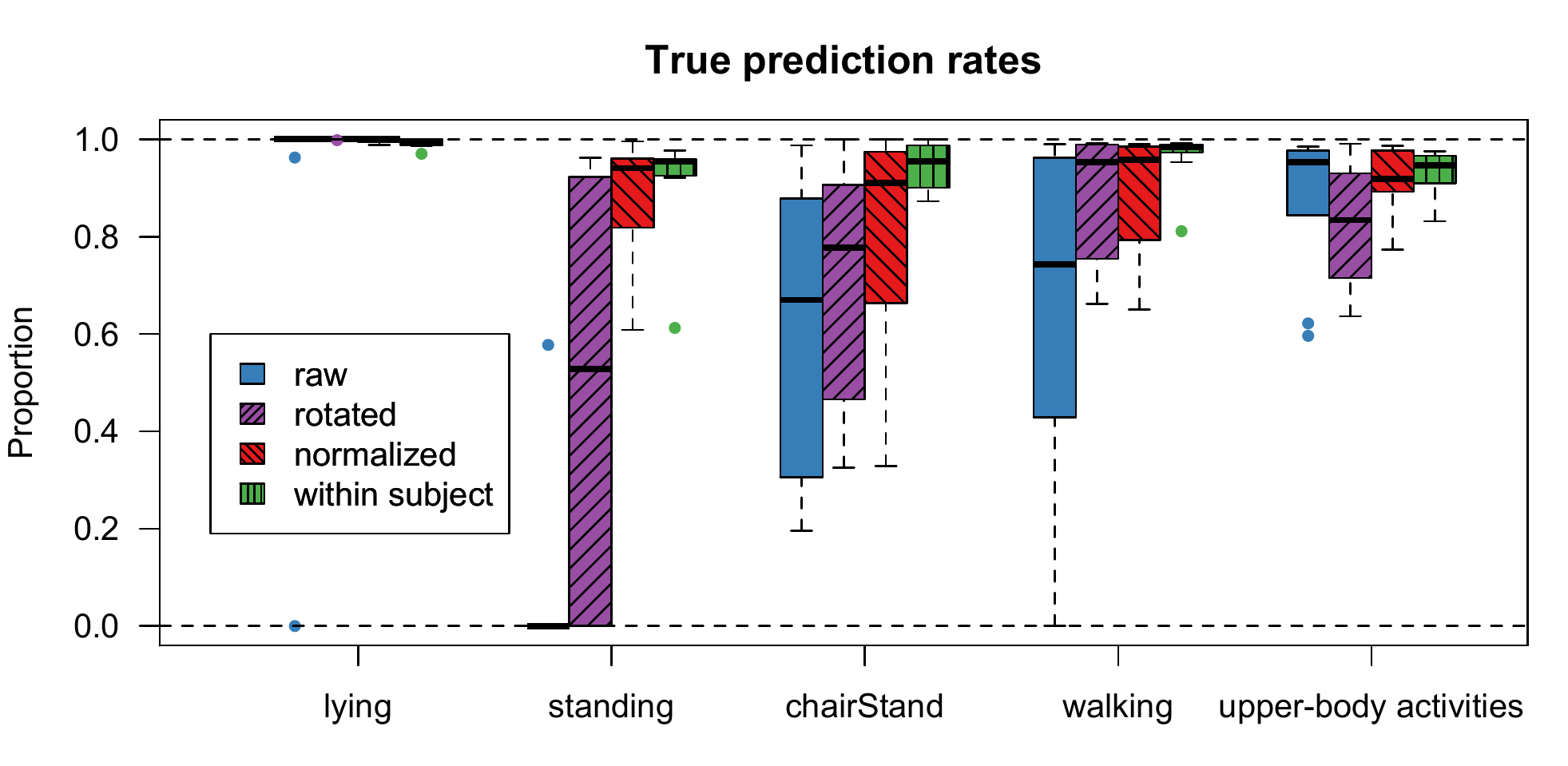}\\
\includegraphics[width=6in,  angle=0]{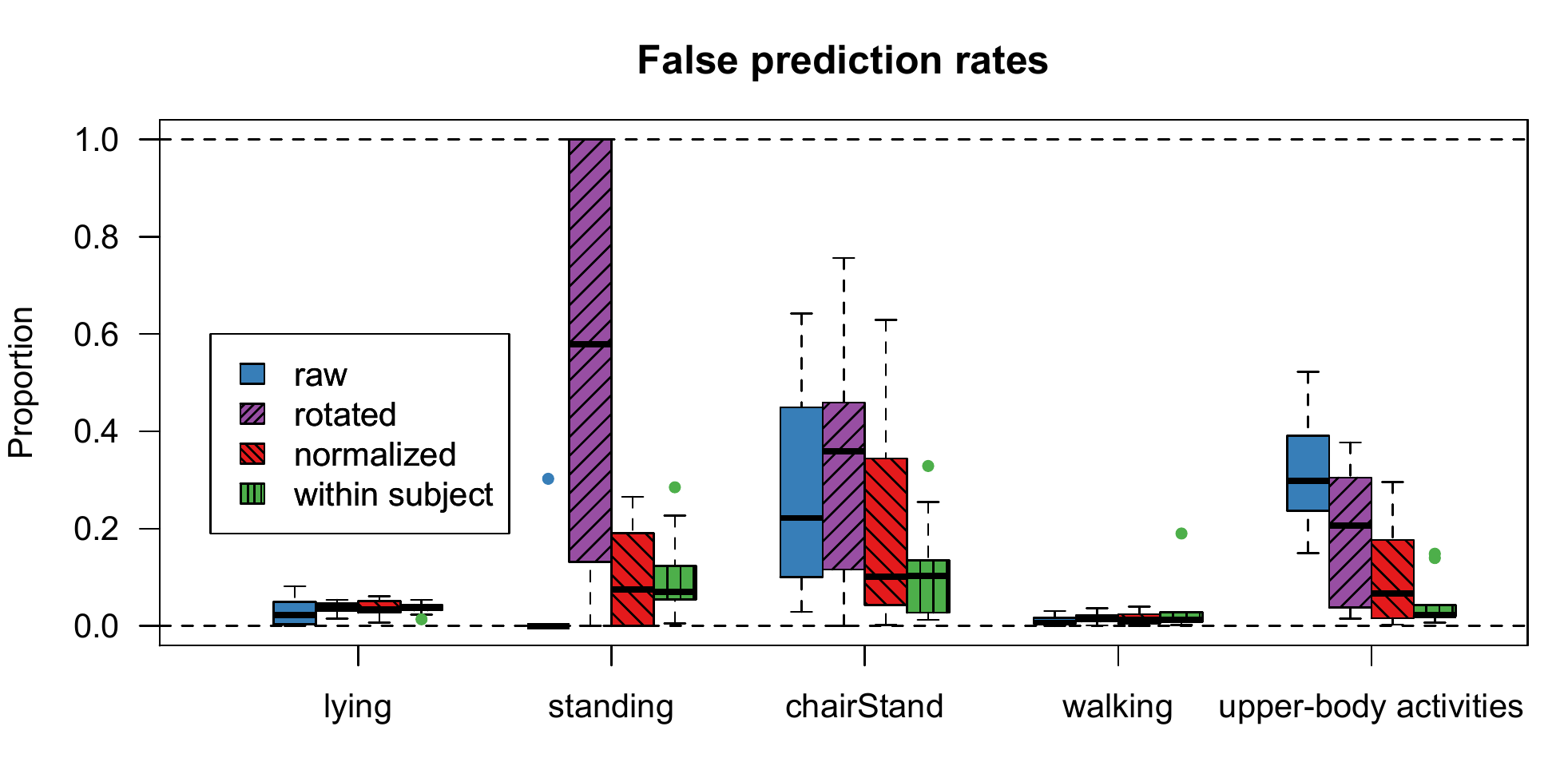}
\end{array}$
\end{center}
\caption{\label{fig_apr}True prediction and false prediction rates for various cases. The term ``raw" means the population-level method is applied
to the raw data; ``rotated" means  the population-level method is applied
to the rotated data without further translation; ``normalized" means the population-level method is applied
to the rotated and translated data; ``within subject" means the within-subject movelet method is applied to each subject. }
\end{figure}

\begin{figure}[htp]
\centering
\includegraphics[width = 6in, angle=0]{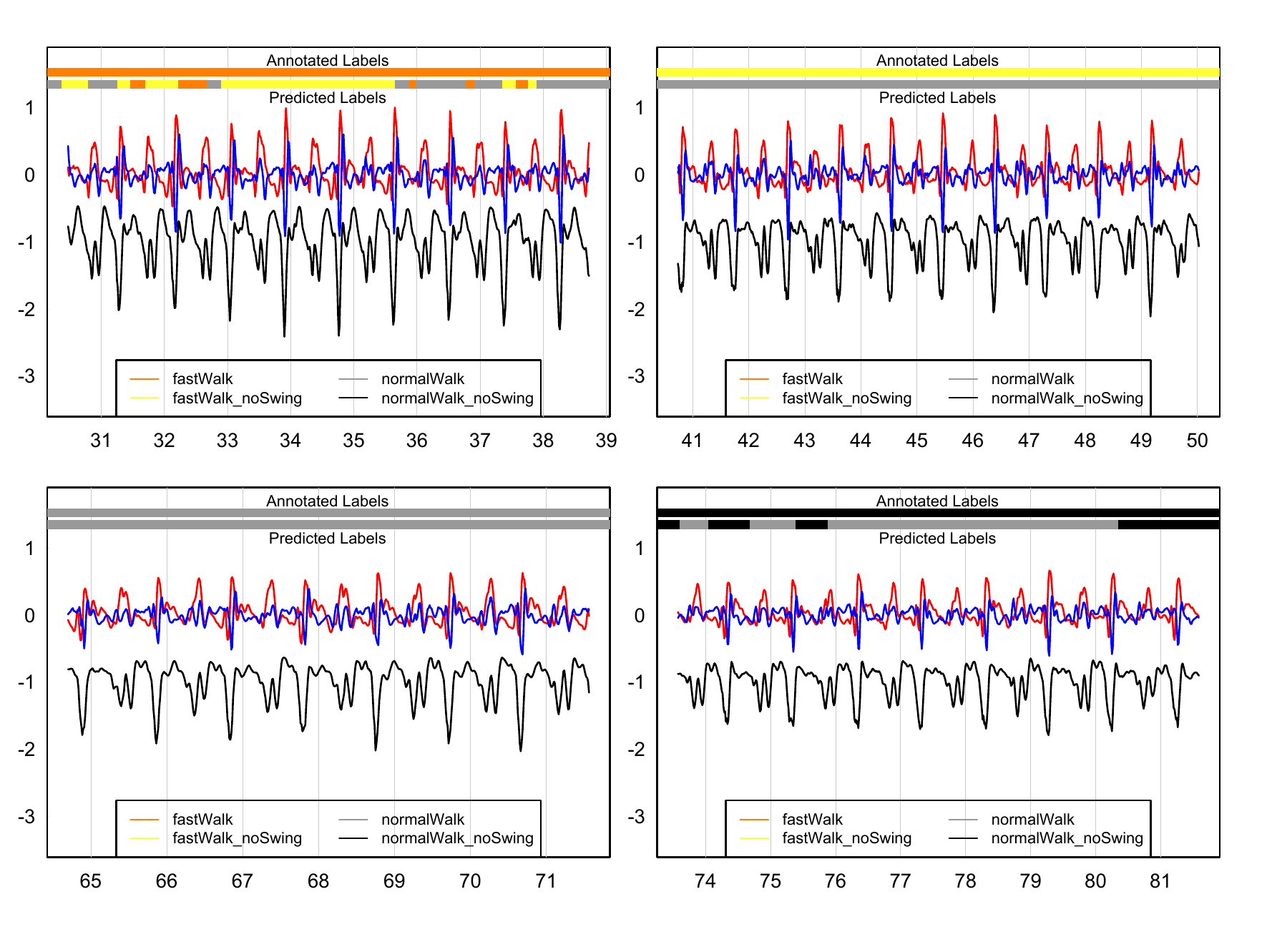}
\caption{\label{plots_walk} Prediction results of one subject's four types of walking. The top panels display data for normalWalk (left column) and normalWalk\_noSwing (right column), the bottom panels display data for fastWalk (left column) and fastWalking\_noSwing (right column). The activity types can also be distinguished by the annotated labels in each plot.}
\end{figure}
\section{Discussion}\label{sec::discuss}
We have proposed a movelet-based method that could predict activities types across subjects. 
Compared to feature extraction-based methods, a significant advantage of the movelet-based methods is that they can achieve high prediction accuracy at sub-second level. We have found in free-living data that walking 2 or 3 steps is common in older individuals. 
To accurately quantify how much time an older person spends on walking, an important biomarker for elderly's health, sub-second labels will be required to capture very short walking periods.

The population-level  movement prediction is a non-trivial step forward  compared to the subject-by-subject methods in \cite{Bai:12} and \cite{He:14}.
Indeed, the methods in \cite{Bai:12} and \cite{He:14} require
training data for all subjects, which is likely unavailable in large epidemiological studies. Moreover, these methods do not consider normalization, which is a crucial problem when devices are worn for many days, are taken off and put back on, and may be subject to unknown movements relative to the body they are attached to. Our proposed methods, require  training data for  a subset of all subjects.

The  data analyzed here were collected in a research lab and  provide only a partial view of the heterogenous activities individuals perform in  free-living environments. It remains unclear how in-lab data prediction methods perform in real life environments. Nonetheless, with the approaches introduced here, we are  mildly optimistic about resolving this much harder problem.

The methods that we proposed here can be developed further. For example,  using the data from all three accelerometers instead of just the hip, may provide better movement recognition of upper body activities. Smoothing the movelets may further reduce the noise in the distance metric. Finally, many movements may actually be quite ambiguously defined. For example, a reaching arm movement could equally well correspond to ``dealing cards", ``placing a plate in drawer", ``eating" or other qualitatively defined movements. Thus, for quantitative research we may need to move to more accurate definitions of movement. Those definitions are likely to be closer in nature to ``movelets" than to current non-quantifiable definitions. This is counter-intuitive and contrary to the way data are currently collected and labeled. However, learning the language of movement will most likely require a careful analysis of the observed data and decomposition into its building blocks. Pairing accelerometry with video cameras, smart phone apps, and other health monitors has the potential to fundamentally change the way we measure health.

\section*{Acknowledgements}
This research was funded in part by the Intramural Research Program of the National Institute on Aging. Xiao, He, and Crainiceanu were supported by Grant Number R01NS060910 from the National Institute of Neurological Disorders and Stroke. This work was also supported by the National Institute of Health through the funded Study of Energy and Aging Pilot (RC2AG036594), Pittsburgh Claude D. Pepper Older American's Independence Center Research Registry (NIH P30 AG024826), a National Institute on Aging Professional Services Contract (HHSN271201100605P), a University of Pittsburgh Department of Epidemiology Small Grant, and the National Institute on Aging Training Grant (T32AG000181). This work represents the opinions of the researchers and not necessarily that of the granting organizations. 
\vspace*{-8pt}

\appendix
\small

\section*{Appendix A: A Test for Systematic Bias}
Let $\mathbf{X}\in \mathbb{R}^3$ be the acceleration vector at an observation point when the subject is standing still.
Suppose  that
$\mathbf{X}$ follows a multivariate normal
 distribution with mean $\boldsymbol{\mu}$ and covariance
$\mathbf{\Sigma}$. Let $\mathbf{X}_1,\ldots,\mathbf{X}_n$ be i.i.d.\ copies
of $\mathbf{X}$. Then $\bar{\mathbf{X}} = n^{-1}\sum_{i=1}^n\mathbf{X}_i$
is normal with mean $\boldsymbol{\mu}$ and covariance $\mathbf{\Sigma}/n$. Testing if there is systematic bias in the
observations is  to testing $\|\boldsymbol{\mu}\| = 1$ where $\|\cdot\|$ denotes the Euclidean norm. We consider the  testing
statistic  $\|\bar{\mathbf{X}}\|^2$,  which has mean
$\|\boldsymbol{\mu}\|^2 + \tr(\mathbf{\Sigma})/n$ and variance $\text{var}(\|\bar{\boldsymbol{X}}\|^2)$. Here $\tr(\cdot)$ denotes
the trace of a square matrix, i.e., the sum of the diagonal entries.
The derivation of the variance term is more involved. We
let $\mathbf{O}\mathbf{D}\mathbf{O}^T$ be the eigendecomposition of $\mathbf{\Sigma}$, where $\mathbf{O}$ is 
an orthogonal matrix with $\mathbf{O}^T\mathbf{O} = \mathbf{O}\mathbf{O}^T = \mathbf{I}_3$ and $\mathbf{D}$ is a diagonal matrix with the
diagonal entries $d_1$, $d_2$ and $d_3$. Now let $\mathbf{Y} = (Y_1, Y_2, Y_3)^T =  \mathbf{O}^T\bar{\mathbf{X}}$, then $\mathbf{Y}$ is 
normal with mean $\boldsymbol{\mu}_y = (\mu_{y1},\mu_{y2},\mu_{y3})^T = \mathbf{O}^T\boldsymbol{\mu}$ and 
covariance matrix $\mathbf{D}/n$. It is easy to show that 
$$
\text{var}(\|\bar{\mathbf{X}}\|^2) = \text{var}(\|\mathbf{Y}\|^2) = \sum_{k=1}^3 \text{var}(Y_k^2) = \sum_{k=1}^3 (6\mu_{yk}^2 d_k/n + 3d_k^2/n^2).
$$
Hence 
\begin{align*}
\text{var}(\|\bar{\mathbf{X}}\|^2) 
& = \sum_{k=1}^3 (6\mu_{yk}^2 d_k/n + 3d_k^2/n^2)\\
& =  \frac{6}{n} \boldsymbol{\mu}_y^T\mathbf{D}\boldsymbol{\mu}_y  + \frac{3}{n^2}\tr(\mathbf{\Sigma}^2)\\
& = \frac{6}{n} \boldsymbol{\mu}^T\mathbf{O}\mathbf{D}\mathbf{O}^{\prime}\boldsymbol{\mu}+  \frac{3}{n^2}\tr(\mathbf{\Sigma}^2)\\
& = \frac{6}{n} \boldsymbol{\mu}^T \mathbf{\Sigma}\boldsymbol{\mu}+  \frac{3}{n^2}\tr(\mathbf{\Sigma}^2).
\end{align*}
By the central limit theorem, 
$
\frac{\|\bar{\mathbf{X}}\|^2-\|\boldsymbol{\mu}\|^2-\tr(\mathbf{\Sigma})/n}{\sqrt{\text{var}(\|\bar{\mathbf{X}}\|^2) }}
$
is approximately normal. Then an $\alpha$-level rejection region for testing $\|\boldsymbol{\mu}\| = 1$ is given by
$$
\left |\|\bar{\mathbf{X}}\|^2-1 \right| > z_{\alpha/2}\sqrt{\frac{6}{n} \boldsymbol{\mu}^T \mathbf{\Sigma}\boldsymbol{\mu}}
$$
Note that we dropped the term $\tr(\mathbf{\Sigma})/n$ in the numerator and the term $3\tr(\mathbf{\Sigma}^2)/n^2$ in the denominator as they are of smaller order than $\|\boldsymbol{\mu}\|^2$ and $6\boldsymbol{\mu}^T \mathbf{\Sigma}\boldsymbol{\mu}/n$, respectively. The term $\boldsymbol{\mu}^T \mathbf{\Sigma}\boldsymbol{\mu}$ is unknown and needs to be estimated. Since under the null hypothesis that
$\|\boldsymbol{\mu}\|=1$ we can derive $\boldsymbol{\mu}^T \mathbf{\Sigma}\boldsymbol{\mu} \leq \|\mathbf{\Sigma}\|_{op}$, where $\|\cdot\|_{op}$ is the operator norm of a matrix,  we use instead
a conservative rejection region
$$
\left |\|\bar{\mathbf{X}}\|^2-1 \right| > z_{\alpha/2}\sqrt{\frac{6}{n} \|\hat{\mathbf{\Sigma}}\|_{op}},
$$
where $\hat{\mathbf{\Sigma}}$ is the sample covariance matrix from the sample $\{\mathbf{X}_1,\ldots, \mathbf{X}_n\}$. We use $\alpha = 0.05$ so that
$z_{\alpha/2} = 1.96$. We display the value of the term
$$
T = \frac{\left |\|\bar{\mathbf{X}}\|^2-1 \right|}{\sqrt{\frac{6}{n} \|\hat{\mathbf{\Sigma}}\|_{op}}}
$$
for all subjects in Table~\ref{web-tab}. The results show that except for subjects 12 and 13, the null hypothesis of $
\|\boldsymbol{\mu}\|$ is always rejected.

\begin{table}
\caption{\label{web-tab}Testing statistic $T$ for the 20 subjects}
\centering
\begin{tabular}{|c|c|}
\hline
Subject&$T$\\\hline
1&209.38\\
2&76.11\\
3&473.84\\
4&103.94\\
5&71.17\\
6&209.28\\
7&183.57\\
8&84.22\\
9&134.98\\
10&365.15\\
11&228.35\\
12&0.65\\
13&1.93\\
14&22.47\\
15&9.59\\
16&191.47\\
17&83.44\\
18&11.68\\
19&11.91\\
20&165.02\\\hline
\end{tabular}
\end{table}

\section*{Appendix B: Derivation of Rotation Matrices}

\begin{lem}
Let $\mathbf{a}_1$ and $\mathbf{a}_2$ be two vectors in $\mathbb{R}^3$ and $\mathbf{a}_1\times\mathbf{a}_2 \neq 0$. Let 
\begin{align*}
\mathbf{b}_1 &= \frac{\mathbf{a}_1}{\|\mathbf{a}_1\|_2},\\
\mathbf{b}_2 &= \frac{\mathbf{a}_2-(\mathbf{a}_2^T\mathbf{b}_1)\mathbf{b}_1}{\|\mathbf{a}_2-(\mathbf{a}_2^T\mathbf{b}_1)\mathbf{b}_1\|_2}.
\end{align*}
Then $\mathbf{b}_1$ and $\mathbf{b}_2$ are two unit vectors and are orthogonal to each other. We write $\mathbf{a}_1$ and $\mathbf{a}_2$ as
\begin{align*}
\mathbf{a}_1 &= c_1 \mathbf{b}_1,\\
\mathbf{a}_2 & = c_2\mathbf{b}_1 + c_3\mathbf{b}_2.
\end{align*}
Let
$$
\mathbf{R}^* =  \arg \min_{\mathbf{R}\in \mathbb{R}^{3\times 3}:\,\, \mathbf{R}^T=\mathbf{R}^{-1}  \text{ and } \mathbf{e}_3^T\mathbf{R}(\mathbf{a}_1\times \mathbf{a}_2)>0}\left(\|\mathbf{R}\mathbf{a}_1 + \mathbf{e}_1\|_2^2 +\|\mathbf{R}\mathbf{a}_2 + \mathbf{e}_2\|_2^2 \right).
$$
Then $\mathbf{R}^*$ is unique with the expression
$$
\left(\begin{array}{ccc}
\cos(\theta)&-\sin(\theta)&0\\
\sin(\theta)&\cos(\theta)&0\\
0&0&1
\end{array}\right)^T[\mathbf{b}_1,\mathbf{b}_2,\mathbf{b}_1\times\mathbf{b}_2]^T,
$$
where 
\begin{align*}
\cos(\theta) &= -\frac{c_1+c_3}{\sqrt{(c_1+c_3)^2 + c_2^2}},\\
\sin(\theta) &= \frac{c_2}{\sqrt{(c_1+c_3)^2 + c_2^2}}.
\end{align*}
\end{lem}
\begin{proof}
For an arbitrary rotation matrix $\mathbf{R}$, $\mathbf{R}^T = \mathbf{R}^{-1}$ is also a rotation matrix. Hence $\mathbf{R}^T\mathbf{e}_1$ and $\mathbf{R}^T\mathbf{e}_1$ remain orthogonal unit vectors. There exists an $\theta\in[0,\pi]$ such that
\begin{equation}
\label{lem_eq1}
\begin{split}
\mathbf{R}^T\mathbf{e}_1 &= \cos(\theta)\mathbf{b}_1 + \sin(\theta)\mathbf{b}_2,\\
\mathbf{R}^T\mathbf{e}_2 &= -\sin(\theta)\mathbf{b}_1 + \cos(\theta)\mathbf{b}_2.
\end{split}
\end{equation}
It follows that 
\begin{align*}
&\|\mathbf{R}\mathbf{a}_1 + \mathbf{e}_1\|_2^2 +\|\mathbf{R}\mathbf{a}_2 + \mathbf{e}_2\|_2^2\\
&=\|\mathbf{a}_1 +\mathbf{R}^T \mathbf{e}_1\|_2^2 +\|\mathbf{a}_2 + \mathbf{R}^T\mathbf{e}_2\|_2^2\\
&=\|(c_1 + \cos(\theta))\mathbf{b}_1 + \sin(\theta)\mathbf{b}_2\|_2^2 + \|(c_2 -\sin(\theta))\mathbf{b}_1 + (c_3 + \cos(\theta))\mathbf{b}_2\|_2^2\\
&=(c_1 + \cos(\theta))^2 + \sin^2(\theta) + (c_2-\sin(\theta))^2 + (c_3 + \cos(\theta))^2\\
&= 2 + c_1^2 +c_2^2 + c_3^2 +  2(c_1+c_3)\cos(\theta) - 2c_2\sin(\theta).
\end{align*}
Therefore, $\|\mathbf{R}\mathbf{a}_1 + \mathbf{e}_1\|_2^2 +\|\mathbf{R}\mathbf{a}_2 + \mathbf{e}_2\|_2^2$ is minimized if
$\cos(\theta) = -(c_1+c_3)/\sqrt{(c_1+c_3)^2 + c_2^2}$ and $\sin(\theta) =c_2/\sqrt{(c_1+c_3)^2 + c_2^2}$.
By~(\ref{lem_eq1}),
\begin{align*}
\mathbf{R}^T\mathbf{e}_3& = \mathbf{R}^T(\mathbf{e}_1\times\mathbf{e}_2)\\
&=(\mathbf{R}^T\mathbf{e}_1)\times(\mathbf{R}^T\mathbf{e}_2)\\
&=\mathbf{b}_1\times\mathbf{b}_2.
\end{align*}
Then 
$$
\mathbf{R}^T[\mathbf{e}_1,\mathbf{e}_2,\mathbf{e}_3] 
= [\mathbf{b}_1,\mathbf{b}_2,\mathbf{b}_1\times\mathbf{b}_2]\left(\begin{array}{ccc}
\cos(\theta)&-\sin(\theta)&0\\
\sin(\theta)&\cos(\theta)&0\\
0&0&1
\end{array}\right).
$$
It's easy to verify that $\mathbf{R}^{-1} = \mathbf{R}^T$ and the proof is complete.
\end{proof}

\bibliographystyle{apalike}
\small
\bibliography{movement_prediction.bib}

\end{document}